\documentclass[runningheads]{llncs}
\usepackage{graphicx}
\usepackage{amsmath}
\usepackage{amssymb}
\usepackage{ascmac}
\usepackage{mathbbol}
\usepackage[bb=boondox]{mathalfa}
\usepackage{bbold}
\usepackage{here}
\usepackage{comment}
\usepackage{enumerate}
\usepackage{color}

\usepackage{physics}
\usepackage[all]{xy}
\usepackage{tikz}
\usetikzlibrary{automata,positioning}

\newcommand{\bqed}{\hfill$\blacksquare$}

\begin{document}

\title{Positive Varieties of Lattice Languages}
\titlerunning{Positive Varieties of Lattice Languages}
%
%\titlerunning{Abbreviated paper title}
% If the paper title is too long for the running head, you can set
% an abbreviated paper title here
%
\author{Yusuke Inoue \and
Yuji Komatsu}
\authorrunning{Y. Inoue and Y. Komatsu}
% First names are abbreviated in the running head.
% If there are more than two authors, 'et al.' is used.
%
\institute{Nagoya University}
\maketitle              % typeset the header of the contribution

\begin{abstract}
While a language assigns a value of either `yes' or `no' to each word, a lattice language assigns an element of a given lattice to each word. 
An advantage of lattice languages is that joins and meets of languages can be defined as generalizations of unions and intersections. 
This fact also allows for the definition of positive varieties\textemdash classes closed under joins, meets, quotients, and inverse homomorphisms\textemdash of lattice languages. 
In this paper, we extend Pin's positive variety theorem, proving a one-to-one correspondence between positive varieties of regular lattice languages and pseudo-varieties of finite ordered monoids.
Additionally, we briefly explore algebraic approaches to finite-state Markov chains as an application of our framework.
\end{abstract}

\section{Introduction}
Multi-valued logic, which includes options such as `unknown' in addition to `yes' and `no', has been studied for a long time.
There have also been several studies on {\it multi-valued automata} and {\it multi-valued languages} as natural counterparts in formal language theory \cite{bruns2004model,kupferman2007lattice}. 
In multi-valued automata with a set of values $V$, each state is assigned a value (or color) in $V$ instead of `accept' or `reject'. 
Therefore, the language accepted by such an automaton is a mapping $L:\Sigma^*\to V$ rather than a subset $L:\Sigma^*\to \{{\rm yes, no}\}$.

As mentioned in \cite{bruns2004model,kupferman2007lattice}, the significant case is that of {\it lattice languages}, where the set of values $V$ forms a lattice $\Lambda=(V,\lor,\land)$. 
Note that the usual languages can be viewed as lattice languages over the Boolean lattice $\mathbb{B}=(\{0,1\},\lor,\land)$.
In addition to this, we can consider the lattice $\Lambda_{\rm ter}=\{0<1/2<1\}$ that includes $1/2$=`unknown', the lattice $\Lambda_{\rm fuz}=[0,1]$ corresponding to fuzzy logic, and so on (see \cite{bruns2004model} for more examples).
One reason lattice languages are significant is that we can define the join and meet operations for lattice languages, which are natural extensions of the union and intersection.
In particular, if $\Lambda$ is distributive, then lattice languages can be considered as weighted languages by regarding $\Lambda$ as a semiring \cite{droste2009handbook}.
Also, any multi-valued language defined on $V$ can be treated as a lattice language by replacing its values with the lattice of the powerset $(\mathcal{P}(V),\cup,\cap)$. 
Therefore, properties of lattice languages are worth investigating.

Monoids are very powerful tools for studying regular languages. 
In particular, the Eilenberg's variety theorem \cite{eilenberg1974automata}, which establishes a one-to-one correspondence between {\it varieties} of regular languages and {\it pseudo-varieties} of finite monoids, is one of the most outstanding results in such studies.
A variety of languages is a class of languages closed under Boolean combinations (unions, intersections, and complements), left and right quotients, and inverse homomorphisms.
Notably, Boolean combinations play important roles in the proof of the theorem.
Let us now consider whether the variety theorem can be extended to lattice languages.
As we mentioned in the previous paragraph, we can define joins and meets for lattice languages; however, in contrast, we cannot define an operation corresponding to the complement.
Thus, the concept of {\it positive variety} introduced by Pin \cite{pin1995variety} can be helpful for our problem. 
Positive varieties, which do not require closure under complements, serve as alternatives to varieties, and they correspond to pseudo-varieties of ordered monoids.
Following the idea of positive varieties, we can prove the variety theorem for lattice languages. 
However, due to the generalization to lattices, the proof is more complex and technical compared to the original proof in \cite{pin1995variety}.

As the main part of this paper, we prove a one-to-one correspondence between positive varieties of lattice languages and pseudo-varieties of ordered monoids.
The definitions and proofs follow the same approach as in \cite{pin1995variety}, but some technical adjustments are required.
First, in \cite{pin1995variety}, the concept of \emph{order ideals} is introduced as downward-closed subsets that behave like accepting states. 
However, when considering a lattice, subsets are not sufficient to characterize lattice languages.
Therefore, we need to extend order ideals to \emph{order-preserving colorings}.
This is a key idea of our framework.
Also, we require that the definition of a positive variety be closed under certain homomorphisms, which are called $\Lambda$-{\it morphisms}. 
As we will see in Section~3, without this modification, it is impossible to establish a one-to-one correspondence between positive varieties and pseudo-varieties. 

In this paper, in addition to proving the variety theorem, we briefly discuss its potential applications. 
Specifically, we claim that lattice languages and ordered monoids can contribute to the study of finite-state Markov chains. 
Historically, algebraic approaches to Markov chains have been explored such as random walks on groups \cite{saloff2004random} and semigroup-theoretic analysis of stationary distributions \cite{rhodes2019unified}. 
Building on this context, we outline how lattice languages can provide meaningful semantics for Markov chains and help describe specific classes of Markov chains.

%In the study of finite-state Markov chains, groups and monoids have been extensively used. However, in such studies, important algebraic tools, such as homomorphisms and varieties, have yet to be fully utilized. One reason for this is that in existing algebraic approaches, groups and monoids do not necessarily reflect the original semantics of Markov chains. In the second part of this paper, we demonstrate that lattice languages can contribute to solving these issues and, furthermore, that the pseudo-variety of ordered monoids can effectively explain certain classes of Markov chains. However, due to space limitations, this section will remain somewhat informal.

%This paper presents the following main contributions: (i) we extend Eilenberg's correspondence to lattice languages within the framework of positive varieties, and (ii) we demonstrate the connection between Markov chains and ordered monoids as a potential application of our result.

This paper is structured as follows. In Section~2, we introduce basic concepts such as ordered monoids, order-preserving colorings, and lattice languages.
In Section~3, we introduce syntactic ordered monoids and prove the variety theorem.
In addition, we explain the relationship with the usual (positive) varieties using some examples.
Finally, in Section~4, we explain how lattice languages and ordered monoids can provide useful algebraic approaches to Markov chains.
Some proofs have been moved to the appendix due to space limitations.

\section{Preliminaries}
A (partial) order is a relation $\le$ on a set that is reflexive, antisymmetric, and transitive.
For two ordered sets $(S_1,\le_1)$ and $(S_2,\le_2)$,
a mapping $\varphi:S_1\to S_2$ is {\it order-preserving} if $s\le_1 s'$ implies $\varphi(s)\le_2 \varphi(s')$ for each $s,s'\in S_1$.
Clearly, the composition of two order-preserving mappings is also order-preserving.

A {\it lattice} is an ordered set in which any two elements have a unique least upper bound and a unique greatest lower bound.
A lattice $\Lambda$ can be considered as an algebra with two binary operations the {\it join} $\lor$ and the {\it meet} $\land$ where $\lambda_1 \lor \lambda_2$ (resp. $\lambda_1\land \lambda_2$) is the least upper bound (resp. greatest lower bound) of $\lambda_1$ and $\lambda_2$ in $\Lambda$.
A lattice $\Lambda$ is \emph{complete} if every subset $I\subseteq \Lambda$ has both the least upper bound $\bigvee I$ and the greatest lower bound $\bigwedge I$.
For a complete lattice $\Lambda$, let $\mathbb{1}=\bigvee \Lambda$ and $\mathbb{0}=\bigwedge \Lambda$.
For instance, for any set $V$, the powerset $(\mathcal{P}(V),\subseteq)$ forms the complete lattice $(\mathcal{P}(V),\cup,\cap)$ with $\mathbb{1}=V$ and $\mathbb{0}=\emptyset$. 
In particular, the Boolean lattice $\mathbb{B}=(\{0,1\},\lor,\land)$ is isomorphic to the case that $V$ is a singleton set.
Note that every finite lattice is complete.
In this paper, we mainly consider the case that $\Lambda$ is complete.
Additionally, we assume that $\Lambda$ is non-trivial, i.e., $\mathbb{0}\ne\mathbb{1}$.

For a lattice $\Lambda$, an order-preserving mapping $\alpha:\Lambda\to \Lambda$ is called a $\Lambda$-{\it morphism}.
For example, for any element $\lambda\in \Lambda$, the constant mapping ${\rm cons}(\lambda):\Lambda\to \Lambda$ defined as ${\rm cons}(\lambda)(x)=\lambda$ for each $x\in \Lambda$ is a $\Lambda$-morphism.

\subsection{Ordered Monoids}
An {\it ordered monoid} is a pair $(M,\le)$ of a monoid $M$ and an order $\le$ on $M$ satisfying for every $u,v,m_1,m_2\in M$, $m_1\le m_2$ implies $um_1v\le um_2v$.
We simply write an ordered monoid $M$ instead of $(M,\le)$ when the order $\le$ is clear from the context.
A {\it morphism of ordered monoids} is an order-preserving monoid morphism.
A morphism $\varphi:(M_1,\le_1)\to (M_2,\le_2)$ is an {\it isomorphism} if there is a morphism $\psi:(M_2,\le_2)\to (M_1,\le_1)$ such that both of $\varphi\circ \psi$ and $ \psi\circ \varphi$ are identity.
As is mentioned in \cite{pin1995variety}, a bijective morphism is not necessarily an isomorphism.
\begin{comment}
\begin{example}\label{ex-ordered_monoids}
Let $R_2=\{1,r_1,r_2\}$ be the right-zero monoid, that is, $r_1\cdot r_2=1$ if $r_1=r_2=1$, and $r_1\cdot r_2=r_2$ otherwise.
Then, $R_2$ with the order $\le$ defined as $\le=\{(1,1),(r_1,r_1),(r_2,r_2),(1,r_1),(1,r_2)\}$ forms an ordered monoid.
In general, the right-zero monoid $R_n=\{1,r_1,\ldots,r_n\}$ with $n$ right-zero elements $r_1,\ldots,r_n$ forms an ordered monoid with $\le=\{(1,1)\}\cup\{(r_i,r_i)\}_i\cup\{(1,r_i)\}_i$.\bqed
\end{example}
\end{comment}

An ordered monoid $(M_1,\le_1)$ is a \emph{submonoid} of $(M_2,\le_2)$ if $M_1$ is an (unordered) submonoid of $M_2$, and $\le_1$ is the restriction of $\le_2$ to $M_1$.
%Note that submonoids are not corresponding to injective morphisms of ordered monoids.
An ordered monoid $(M_1,\le_1)$ is a \emph{quotient} of $(M_2,\le_2)$ if there is a surjective morphism $\varphi:(M_2,\le_2)\to (M_1,\le_1)$.
A quotient of a submonoid is also called a {\it divisor}.
For a family of ordered monoids $\{(M_i,\le_i)\}_i$, the {\it direct product} $\prod_{i}(M_i,\le_i)$ is the ordered monoid on the product of monoids $\prod_i M_i$ with the order $\le$ defined as $(u_i)_i \le (v_i)_i \iff \forall i.(u_i \le_i v_i)$.

Let $\Sigma$ be a finite alphabet, and let $\Sigma^*$ be the free monoid on $\Sigma$. 
The following properties of the free (ordered) monoid are shown in \cite{pin1995variety}.
\begin{proposition}[Proposition 2.5, \cite{pin1995variety}]\label{prop-univ_free_1}
For any mapping $\varphi:\Sigma\to M$ with an ordered monoid $(M,\le)$, there is a unique morphism $\overline{\varphi}:(\Sigma^*,=)\to (M,\le)$ such that $\varphi(a)=\overline{\varphi}(a)$ for each $a\in \Sigma$.\bqed
\end{proposition}
\begin{proposition}[Corollary 2.6, \cite{pin1995variety}]\label{prop-univ_free_2}
Let $\eta:(\Sigma^*,=)\to (M_1,\le_1)$ be a morphism and $\varphi:(M_2,\le_2)\to (M_1,\le_1)$ be a surjective morphism. There is a morphism $\psi:(\Sigma^*,=)\to (M_2,\le_2)$ such that $\eta=\phi\circ \psi$.\bqed
\end{proposition}
\begin{comment}
    \begin{equation*}
  \xymatrix{
  {(\Sigma^*,=)}\ar@{..>}[r]^-{ \exists \psi}\ar@{->}[rd]_-{\eta}
  &{(M_2,\le_2)}\ar@{->}[d]^-{\varphi}\\
  &{(M_1,\le_1)} ~.\!\!
  }
\end{equation*}
commutes.
\end{comment}

\subsection{Order-Preserving Colorings}
Let $\Lambda$ be a complete lattice. 
For an ordered monoid $(M,\le)$, an order-preserving mapping $P:M\to \Lambda$ is called an {\it order-preserving coloring} (or {\it op-coloring}) on $(M,\le)$.

\begin{remark}\label{rem-ideal}
A subset $I$ of an ordered monoid $(M,\le)$ is called an {\it order ideal} in \cite{pin1995variety} if $x\le y$ and $y\in I$ implies $x\in I$ for all $x,y\in M$. It is clear that $I\subseteq M$ is an order ideal iff $I=P^{-1}(0)$ with an op-coloring $P:M\to \mathbb{B}$. In this sense, the concept of op-colorings is a natural extension of order ideals.\bqed
\end{remark}

For a family $\{P_i\}_i:M\to \Lambda$ of op-colorings on an ordered monoid $(M,\le)$, define the {\it join} $\bigvee_{i}P_i:M\to \Lambda$ and the {\it meet} $\bigwedge_{i}P_i:M\to \Lambda$ as the pointwise extensions, that is,\[\textstyle
(\bigvee_{i}P_i)(m)=\bigvee_{i}(P_i(m)),~~(\bigwedge_{i}P_i)(m)=\bigwedge_{i}(P_i(m))
\]
for each $m\in M$.
We can easily show that both of $\bigvee_{i}P_i$ and $\bigwedge_{i}P_i$ are order-preserving.
In addition, let $\{(M_i,\le_i)\}_i$ be a family of ordered monoids, and let $P_i:M_i\to \Lambda$ be op-colorings for each $i$.
Define the {\it product join} $\bigoplus_i P_i$ and the {\it product meet} $\bigotimes_i P_i$ as the mappings from $\prod_i(M_i,\le_i)$ to $\Lambda$ such that \[\textstyle
    (\bigoplus_i P_i)(m)= \bigvee_i P_i(m_i),~~
    (\bigotimes_i P_i)(m)= \bigwedge_i P_i(m_i)
\]
for each $m=(m_i)_i\in \prod_i(M_i,\le_i)$.
\begin{proposition}
Let $\{(M_i,\le_i)\}_i$ be a family of ordered monoids, and let $P_i:M_i\to \Lambda$ be op-colorings.
Then, both of $\bigoplus_i P_i$ and $\bigotimes_i P_i$ are op-colorings.
\end{proposition}
\begin{proof}
Let $m=(m_i)_i$ and $m'=(m_i')_i$ be elements of $\prod_i(M_i,\le_i)$ such that $m\le m'$.
By the definition of the product of ordered monoids, $m_i\le_i m_i'$ for each $i$.
Because $P_i$ is order preserving for each $i$, $P_i(m_i)\le_i P_i(m_i')$ also holds.
Then, $(\bigoplus_i P_i)(m)=\bigvee_i P_i(m_i)\le \bigvee_i P_i(m_i')=(\bigoplus_i P_i)(m')$ and $(\bigotimes_i P_i)(m)=\bigwedge_i P_i(m_i)\le \bigwedge_i P_i(m_i')=(\bigotimes_i P_i)(m')$ hold.
This implies that $\bigoplus_i P_i$ and $\bigotimes_i P_i$ are order-preserving.\qed
\end{proof}

Let $P$ be an op-coloring on $(M,\le)$. For an element $u\in M$, the {\it left quotient} $u\backslash P:M\to \Lambda$ and the {\it right quotient} $P/u:M\to \Lambda$ with $u$ are defined as $u\backslash P(x)=P(ux)$ and $P/u(x)=P(xu)$ for each $x\in M$.
\begin{proposition}\label{prop-quot_is_op}
    For an op-coloring $P$ on $(M,\le)$ and an element $u\in M$, both of $u\backslash P$ and $P/u$ are op-colorings.
\end{proposition}
\begin{proof}
Let $x\le y\in M$. Because the order $\le$ is compatible with the operation of $M$, $ux\le uy$ for any $u\in M$.
Then, $u\backslash P(x)=P(ux)\le P(uy)=u\backslash P(y)$, that is, $u\backslash P$ preserves the order. The same holds for the right quotient.\qed
\end{proof}
\begin{proposition}\label{prop-quot_and_morphism}
    Let $\eta:M_1\to M_2$ be a morphism with ordered monoids $M_1$ and $M_2$.
    Also, let $P$ be an op-coloring on $M_2$.
    For each element $u\in M_1$, $u\backslash (P\circ \eta)= (\eta(u)\backslash P)\circ \eta$ holds.
\end{proposition}
\begin{proof}
    For each $x\in M_1$, we have $u\backslash (P\circ \eta)(x)=(P\circ \eta )(ux)=P(\eta (u)\eta(x))=(\eta(u)\backslash P)( \eta(x))=((\eta(u)\backslash P)\circ \eta)(x)$. (In fact, this proposition holds regardless the order.) \qed
\end{proof}

Let $P$ be an op-coloring on $(M_2,\le_2)$, and let $h:(M_1,\le_1)\to (M_2,\le_2)$ be a morphism of ordered monoids.
We say that $P\circ h$ is the {\it inverse homomorphism} with $h$.
This is because if we regard the condition $P(x)=\mathbb{0}$ as ``$x\in M_2$ is in the subset $\hat P\subseteq M_2$''\footnote{Following \cite{pin1995variety}, we interpret $\mathbb{1}$ as `rejected' and $\mathbb{0}$ as `accepted'. However, by considering the dual of each lattice, this interpretation can be reversed without affecting the substance of the argument.}, then,\[
Ph(x)=\mathbb{0} ~\iff~ h(x)\in \hat P ~\iff ~ x\in h^{-1}(\hat P)
\] holds for each $x\in M_1$. 
It is clear that $Ph$ is order preserving.

\subsection{Lattice Languages}
Let $\Sigma$ be a finite alphabet, and let $\Lambda$ be a complete lattice. 
A $\Lambda$-{\it language} over $\Sigma$ is a mapping $L:\Sigma^*\to \Lambda$.
Note that a $\Lambda$-language is also an op-coloring $L:(\Sigma^*,=)\to \Lambda$.
Therefore, joins, meets, left and right quotients, and inverse homomorphisms are defined in the same way as the previous subsection.
%In addition, for $\Lambda$-languages $L_1,\ldots,L_k$ over possibly different alphabets $\Sigma_1,\ldots,\Sigma_k$, define the join $\bigvee_{i}L_i:(\Sigma_1\cup\cdots\cup \Sigma_k)^*\to \Lambda$ and the meet $\bigwedge_{i}L_i:(\Sigma_1\cup\cdots\cup \Sigma_k)^*\to \Lambda$ by letting $L_i(w)=\mathbb{1}$ if $w\notin (\Sigma_i)^*$ for each $1\le i\le k$.

A $\Lambda$-language $L$ is {\it recognized} by an ordered monoid $(M,\le)$ if there is a morphism $\eta:(\Sigma^*,=)\to (M,\le)$ and an op-coloring $P:(M,\le)\to \Lambda$ such that the diagram
\begin{equation*}
  \xymatrix{
  {(\Sigma^*,=)}\ar@{..>}[r]^-{ \exists\eta}\ar@{->}[d]_-{L}
  &{(M,\le)}\ar@{..>}[ld]^-{\exists P}\\
  {\Lambda}&
  }
\end{equation*}
commutes, that is, $L=P\circ \eta$. In this case, we also say that $L$ is recognized by the triple $(\eta,(M,\le),P)$.
This definition is an extension of the definition of recognition in \cite{pin1995variety} (see also Remark~\ref{rem-ideal}).
A $\Lambda$-language $L$ is {\it regular} if it is recognized by a finite ordered monoid.
The following are basic properties of recognition.

\begin{theorem}\label{th-basic_closure_properties}
Let $\Lambda$ be a complete lattice, and let $L_1$ and $L_2$ be lattice languages over $\Sigma$ recognized by $(M_1,\le_1)$ and $(M_2,\le_2)$, respectively. Then,
\begin{enumerate}[(i)]
    \item the join $L_1\lor L_2$ and the meet $L_1\land L_2$ are recognized by ${(M_1,\le_1)}\times (M_2,\le_2)$,
    \item for a word $u\in \Sigma^*$, the quotients $u\backslash L_1$ and $L_1/u$ are recognized by $(M_1,\le_1)$,
    \item for a language homomorphism $h:(\Sigma')^*\to \Sigma^*$, the inverse homomorphism $L_1h$ is recognized by $(M_1,\le_1)$, and
    \item for a $\Lambda$-morphism $\alpha:\Lambda\to \Lambda$, $\alpha L_1$ is recognized by $(M_1,\le_1)$.
\end{enumerate}
\end{theorem}
\begin{proof}
{\bf (i)} Assume that $L_1$ and $L_2$ are recognized by $(\eta_1,(M_1,\le_1),P_1)$ and $(\eta_2,(M_2,\le_2),P_2)$, respectively.
Let $\langle \eta_1,\eta_2\rangle:\Sigma^*\to (M_1,\le_1)\times (M_2,\le_2)$ be the homomorphism defined as $\langle \eta_1,\eta_2\rangle(w)=(\eta_1(w),\eta_2(w))$ for each $w\in \Sigma^*$.
Then, we can easily show that $L_1\lor L_2$ is recognized by $(\langle \eta_1,\eta_2\rangle,(M_1,\le_1)\times {(M_2,\le_2)},P_1\oplus P_2)$, and $L_1\land L_2$ is recognized by $(\langle \eta_1,\eta_2\rangle,{(M_1,\le_1)}\times {(M_2,\le_2)},P_1\otimes P_2)$.

{\bf (ii)} $u\backslash L_1$ is recognized by $(\eta_1,(M_1,\le_1),\eta_1(u)\backslash P_1)$ because
$u\backslash L_1=u\backslash (P_1\circ \eta_1)=(\eta_1(u)\backslash P_1)\circ \eta_1$ by Proposition~\ref{prop-quot_and_morphism}. The same holds for $L_1/u$.

{\bf (iii and iv)} The equation $\alpha \circ L_1 \circ h= (\alpha P_1)\circ (\eta_1 h) $ holds as in the following commutative diagram:
\begin{equation*}
  \xymatrix{
  {((\Sigma')^*,=)}\ar@{->}[r]^-{ h}
  &{(\Sigma^*,=)}\ar@{->}[r]^-{ \eta_1}\ar@{->}[d]_-{L_1}
  &{(M_1,\le_1)}\ar@{->}[ld]^-{P_1}\\
  &{\Lambda}\ar@{->}[r]_-{ \alpha}&{\Lambda}.
  }
\end{equation*}
Therefore, 
$\alpha L_1  h$ is recognized by $( \eta_1 \circ h,(M_1,\le_1),\alpha\circ P_1)$.\qed
\end{proof}

Next, let us discuss recognition by divisors. The case of quotients is almost the same as in \cite{pin1995variety}, but the case of submonoids requires some additional discussion.
\begin{theorem}\label{th-basic_recognition}
    Let $(M_1,\le_1)$ and $(M_2,\le_2)$ be ordered monoids. 
    For a complete lattice $\Lambda$, let $L$ be a $\Lambda$-language over $\Sigma$ recognized by $(M_1,\le_1)$. Then,
    \begin{enumerate}[(i)]
    \item if $(M_1,\le_1)$ is a quotient of $(M_2,\le_2)$, then $(M_2,\le_2)$ also recognizes $L$,
    \item if $(M_1,\le_1)$ is a submonoid of $(M_2,\le_2)$, then $(M_2,\le_2)$ also recognizes $L$, and
    \item if $(M_1,\le_1)$ is a divisor of $(M_2,\le_2)$, then $(M_2,\le_2)$ also recognizes $L$.
\end{enumerate}
\end{theorem}
\begin{proof}
   Assume that $L$ is recognized by $(\eta,(M_1,\le_1),P)$ with a morphism $\eta:(\Sigma^*,=)\to (M_1,\le_1)$ and an op-coloring $P:M_1\to \Lambda$.
    {\bf (i)} 
    If $(M_1,\le_1)$ is a quotient of $(M_2,\le_2)$, there is a surjective morphism $\varphi:(M_2,\le_2)\to(M_1,\le_1)$. 
    Then, by Proposition~\ref{prop-univ_free_2}, there is a morphism $\psi:(\Sigma^*,=)\to (M_2,\le_2)$ such that $\eta=\phi\circ \psi$.
    Thus, we have that $L$ is recognized by $(\psi,(M_2,\le_2),P\circ \phi)$.

    {\bf (ii)} If $(M_1,\le_1)$ is a submonoid of $(M_2,\le_2)$, then there is an embedding morphism $\phi:(M_1,\le_1)\to (M_2,\le_2)$ such that $\phi(m)=m$ for each $m\in M_1$.
    Define $P':M_2\to \Lambda $ as \[
    P'(m)=\bigvee\{P(x)\mid x\in M_1, x\le_2 m\}.
    \]
    Then, $P'$ is order-preserving because if $m\le_2m'$ holds with some elements $m,m'\in M_2$, we have that $\{P(x)\mid x\in M_1, x\le_2 m\}\subseteq \{P(x)\mid x\in M_1, x\le_2 m'\}$. Therefore,\[
    P'(m)=\bigvee\{P(x)\mid x\in M_1, x\le_2 m\}\le\bigvee\{P(x)\mid x\in M_1, x\le_2 m'\}=P'(m')
    \]
    holds. 
    Also, because $P'(m)=P(m)$ holds if $m\in M_1$, we obtain $P'\circ \phi=P$.
    Consequently, $L=P'\circ \phi\circ \eta$ is recognized by $(\phi\circ \eta,(M_2,\le_2),P')$.
    
    {\bf (iii)} The statement follows from {\bf (i)} and {\bf (ii)}.\qed
\end{proof}

\section{Eilenberg's Correspondence for Lattice Languages}
For a complete lattice $\Lambda$, we say that a class ${\bf L}$ of $\Lambda$-languages is closed under $\Lambda$-morphisms if for any language $L\in \bf L$ and any $\Lambda$-morphism $\alpha:\Lambda\to \Lambda$, it holds that $\alpha L\in {\bf L}$.
A class of regular $\Lambda$-languages is called a {\it positive variety} if it is closed under joins, meets, left and right quotients, inverse homomorphisms, and $\Lambda$-morphisms.
Also a class of finite ordered monoids is called a {\it pseudo-variety} if it is closed under taking divisors and finite direct products.
The goal of this section is to prove the Eilenberg's correspondence:
\begin{claim}
Let $\Lambda$ be a complete lattice.
There is a natural correspondence between positive varieties of $\Lambda$-languages and pseudo-varieties of finite ordered monoids.
\end{claim}
The standard definition of a positive variety requires closure under joins, meets, left and right quotients, and inverse homomorphisms. However, our definition additionally requires closure under $\Lambda$-morphisms.
This is based on the following technical reason: 
Without $\Lambda$-morphisms, a one-to-one correspondence cannot be established. Let us consider the simplest example. 
For a lattice $\Lambda$ with three or more elements, the class of languages ${\bf Cons}=\{{\rm cons}(\lambda):\Sigma^*\to \Lambda\mid \lambda\in \Lambda,\Sigma\text{ is a finite alphabet}\}$ consisting of all constant colorings is clearly closed under joins, meets, left and right quotients, and inverse homomorphisms. 
Also, its subclass ${\bf B}=\{{\rm cons}({\mathbb{0}}),{\rm cons}({\mathbb{1}})\}$ is closed under these operations. (This class corresponds to $\{\emptyset,\Sigma^* \}$ in the usual sense.) 
However, both ${\bf Cons}$ and ${\bf B}$ correspond to the same pseudo-variety of ordered monoids, which consists only of the trivial monoid. 
This leads to a many-to-one correspondence between standard positive varieties and pseudo-varieties. 
Intuitively, ${\bf B}$ is inadequate as a class of $\Lambda$-language because it is inherently a class of $\mathbb{B}$-languages. 
The solution to this problem is taking closure of $\Lambda$-morphisms.
Because all constant coloring can be constructed using $\Lambda$-morphisms, ${\bf B}$ does not form a positive variety in our framework.
Note that the positive variety of $\mathbb{B}$-languages corresponds to the positive variety of languages in the usual sense.
This is because a $\mathbb{B}$-morphism is either ${\rm cons}(0)$, ${\rm cons}(1)$, or the identity mapping.
Therefore, our result is a valid generalization of \cite{pin1995variety}.

\subsection{Syntactic Ordered Monoids}
We first define the syntactic ordered monoids of $\Lambda$-languages.
Let $\Lambda$ be a complete lattice with an order $\le$, and let $L$ be a $\Lambda$-language over $\Sigma$.
Define the preorder $\preceq_L$ on $\Sigma^*$ as\[
w_1 \preceq_L w_2 \iff \forall u,v\in \Sigma^*.\, (L(uw_1v)\le L(uw_2v) )~,
\]
and define the equivalence relation $\simeq_L$ on $\Sigma^*$ as $w_1 \simeq_L w_2 $ iff $(w_1 \preceq_L w_2$ and $w_2 \preceq_L w_1)$.
We call $\simeq_L$ the {\it syntactic congruence} of $L$.
This is indeed a congruence, and thus, the quotient $M_L=\Sigma^*/\simeq_L$ forms a monoid.
The ordered monoid $(M_L,\le_L)$ with the order $\le_L$ induced by $\preceq_L$ is called the {\it syntactic ordered monoid} of $L$.
In addition, the natural surjective morphism $\eta_L:{(\Sigma^*,=)}\to (M_L,\le_L)$ where $\eta_L(w)$ is the equivalence class of $w$ is called the {\it syntactic morphism}
of $L$.

The key property of the syntactic ordered monoid $M_L$ is that it is the `smallest' monoid recognizing $L$. 
The following theorem, which states this, can be shown in almost the same way as in \cite{pin1995variety}.
\begin{theorem}\label{th-minimality_of_synt}
Let $L$ be a $\Lambda$-language over an alphabet $\Sigma$.
If an ordered monoid $(M,\le)$ recognizes $L$, then the syntactic ordered monoid $(M_L,\le_L)$ is a divisor of $(M,\le)$.
\end{theorem}
\begin{proof}
Assume that $L$ is recognized by $(\eta,(M,\le),P)$.
Then, $N={\rm Im}(\eta)$ forms a submonoid of $(M,\le)$ with the restricted order of $\le$ to $N$.
We can show that $M_L$ is an (unordered) quotient of $N$ in the usual way:
Let ${\rm ker}(\eta)$ be the kernel of $\eta$. Then, ${\rm ker}(\eta)$ is a subset of $\simeq_L$, and this implies that the natural surjective mapping $\phi:[w]_{{\rm ker}(\eta)}\mapsto [w]_{\simeq_L}$ is a monoid morphism from $N$ to $M_L$.
Thus, it suffices to show that $\phi$ preserves the order.
Let $m_1\le m_2\in N$, and assume that $m_1=\eta(w_1)$ and $m_2=\eta(w_2)$ with $w_1,w_2\in \Sigma^*$.
Then, for all $u,v\in \Sigma^*$, $\eta(uw_1v)=\eta(u)m_1\eta(v)\le \eta(u)m_2\eta(v)=\eta(uw_2v)$ holds by the assumption that $\le$ is compatible with the operation of $M$.
In addition, because $P$ is order preserving, $L(uw_1v)=P\eta(uw_1v)\le P\eta(uw_2v)=L(uw_2v)\in \Lambda$.
Thus, $w_1\preceq_L w_2$ holds, and we have $\phi(m_1)=[w_1]_{\simeq_L}\le_L [w_2]_{\simeq_L}=\phi(m_2)$.\qed
\end{proof}

\subsection{Ideal Colorings}
Let $\Lambda$ be a lattice, and let $m\in M$ be an element of an ordered monoid $(M,\le)$.
The $m$-{\it ideal coloring} is the op-coloring $\iota[m]:M\to \Lambda$ such that\[
\iota[m](x)=\begin{cases}
\mathbb{0} &\text{if }x\le m,\\
\mathbb{1} &\text{otherwise.}
\end{cases}
\]
Ideal colorings play an important role in the following fact.
\begin{proposition}
   \label{prop-ideal_representation}
Let $\Lambda$ be a complete lattice, and let $(M,\le)$ be an ordered monoid.
For an op-coloring $P:M\to \Lambda$, it holds that $P=\bigwedge_{m\in M}(\iota[m]\lor {\rm cons}(P(m)))$.
\end{proposition}
\begin{proof}
    Let $Q=\bigwedge_{m\in M}(\iota[m]\lor {\rm cons}(P(m)))$.
    We show that $P(x)=Q(x)$ for all $x\in M$:
    \begin{align*}
        Q(x)&= (\bigwedge_{m\in M}(\iota[m]\lor {\rm cons}(P(m))))(x)\\
        &=\bigwedge_{m\in M}(\iota[m](x)\lor {\rm cons}(P(m))(x))\\
        &=\bigwedge_{m\in M}(\iota[m](x)\lor P(m))
    \end{align*}
where $\iota[m](x)\lor P(m)$ is $P(m)$ if $x\le m$, and $\mathbb{1}$ otherwise.
Thus, we have
\[
    Q(x)= \bigwedge_{m\in M,~x\le m}P(m)
    =P(x)\]
because $P(x)\le P(m)$ if $x\le m$. This concludes the proof.\qed
\end{proof}

The following lemma is an essential part of the variety theorem.

\begin{lemma}\label{lem-disjunct_1}
Let $L:\Sigma^*\to \Lambda$ be a regular $\Lambda$-language, and let $M_L$ and $\eta_L$ be the syntactic ordered monoid and the syntactic morphism of $L$.
Then, for every element $m\in M_L$, the language $\iota[m]\circ \eta_L$ can be represented by using joins, meets, quotients, and $\Lambda$-morphisms with $L$.
\end{lemma}
\begin{proof}
    If $m\in M_L$ is the greatest element, then the statement is clear because $\iota[m]\circ \eta_L={\rm cons}(\mathbb{0})\circ L$ holds with the $\Lambda$-morphism ${\rm cons}(\mathbb{0}):\Lambda\to \Lambda$.
    We assume that $m\in M_L$ is not the greatest element.
    Let $P_L:M_L\to \Lambda$ denote the op-coloring such that $L=P_L\circ \eta_L$.
    Pick $y\in M_L$ such that $ y\nleq m$.
    By the definition of the syntactic ordered monoid, 
    there are words $w,w'\in M_L$ such that $P_L(uyu')\nleq P_L(umu')$ with $u=\eta_L(w)$ and $u'=\eta_L(w')$.
    Then, $(u\backslash P_L/u')(y) = P_L(uyu')\nleq P_L(umu') =(u\backslash P_L/u')(m) $.
    Here, let $\alpha_y:\Lambda \to \Lambda$ be the $\Lambda$-morphism defined as\[
\alpha_y(\lambda)=\begin{cases}
\mathbb{0} &\text{if }\lambda\le (u\backslash P_L/u')(m),\\
\mathbb{1} &\text{otherwise.}
\end{cases}
\]
Let $Q_y=\alpha_y\circ (u\backslash P_L/u')$.
Then, we have $Q_y(x)=\mathbb{0}$ if $x\le m$, and $Q_y(x)=\mathbb{1}$ if $x=y$ by the following reasons:
The first part follows from that $x\le m$ implies $(u\backslash P_L/u')(x)\le (u\backslash P_L/u')(m)$ because $(u\backslash P_L/u')$ is order-preserving.
The second part follows because $(u\backslash P_L/u')(y)\nleq (u\backslash P_L/u')(m)$ as we have seen before.
Therefore, $\forall y\nleq m.(Q_y(x)=\mathbb{0})$ if $x\le m$, and $Q_x(x)=\mathbb{1}$ if $x\nleq m$.
Consequently, the ideal coloring $\iota[m]$ can be constructed as $\iota[m]=\bigvee_{y\nleq m}Q_y$.
Then, we have that \begin{align*}
    \iota[m]\circ \eta_L&=(\bigvee_{y\nleq m}Q_y)\circ \eta_L&[\iota[m]=\bigvee_{y\nleq m}Q_y]\\
    &=\bigvee_{y\nleq m}(\alpha_y\circ (u\backslash P_L/v')\circ \eta_L)&[Q_y=\alpha\circ (u\backslash P_L/u')]\\
    &=\bigvee_{y\nleq m}(\alpha_y\circ (\eta_L(w)\backslash P_L/\eta_L(w'))\circ \eta_L)&[u=\eta_L(w),u'=\eta_L(w')]\\
    &=\bigvee_{y\nleq m}(\alpha_y\circ (w\backslash L/w'))&[\text{by Proposition~\ref{prop-quot_and_morphism}]}
\end{align*} holds.
Thus, $\iota[m]\circ \eta_L$ is represented by using joins, meets, quotients, and $\Lambda$-morphisms with $L$.\qed
\end{proof}
We have the following lemma based on this fact.

\begin{lemma}
\label{lem-recog_by_synt}
    Let ${\bf L}$ be a positive variety of $\Lambda$-languages. Let $L_i\in {\bf L}$ be a $\Lambda$-language over $\Sigma_i$ and let $M_{L_i}$ be its syntactic ordered monoid with $1\le i\le k$.
    Then, every $\Lambda$-language recognized by $M=\prod_{1\le i\le k}{M_{L_i}}$ is contained in ${\bf L}$.
\end{lemma}
\begin{proof}
Let $L$ be a $\Lambda$-language over $\Sigma$ recognized by $(\eta,M,P)$.
We show that $L$ is in $\bf L$.
Let $\pi_i:M\to M_i$ denote the $i$-the projection for each $1\le i\le k$, that is, $\pi_i$ is the monoid morphism such that $\pi_i(m)=m_i$ for each $m=(m_1,\ldots,m_k)\in M$.
By Proposition~\ref{prop-ideal_representation}, $P:M\to \Lambda$ can be represented by $
    P=\bigwedge_{m\in M}(\iota[m]\lor {\rm cons}( P(m))).$
Thus, we have that
\begin{align*}
L&=P\circ \eta\\
&=\bigwedge_{m\in M}(\iota[m]\lor {\rm cons}( P(m))) \circ \eta\\
    &=\bigwedge_{m\in M}(\iota[m]\circ \eta \lor {\rm cons}( P(m)))
\end{align*} 
holds. Note that every constant mapping is a $\Lambda$-morphism.
In addition, the op-coloring $\iota[m]\circ \eta:\Sigma^*\to \Lambda$ is equal to $\bigvee_{1\le i\le k}(\iota[m_i]\circ \pi_i\circ\eta)$ with $m=(m_1,\ldots,m_k)$ by the following reasons:
First,  it is clear that both op-colorings take only the values $\mathbb{0}$ or $\mathbb{1}$.
Then,
\begin{align*}
(\iota[m]\circ \eta)(w)=\mathbb{0}&\iff \eta(w)\le (m_1,\ldots,m_k)\\
&\iff \pi_i(\eta(w))\le m_i\text{ for each }1\le i\le k\\
&\iff \iota[m_i](\pi_i\circ \eta(w))=\mathbb{0}\text{ for each }1\le i\le k\\
&\iff \bigvee_{1\le i\le k}\bigl((\iota[m_i]\circ \pi_i\circ\eta)(w)\bigr)=\mathbb{0}
\end{align*} 
holds for each $w\in \Sigma^*$. 
This implies $\iota[m]\circ \eta=\bigvee_i(\iota[m_i]\circ \pi_i\circ\eta)$.
Therefore, $L$ can be represented by using meets, joins, and $\Lambda$-morphisms of languages $\iota[m_i]\circ \pi_i\circ\eta:\Sigma^*\to \Lambda$ with $m=(m_1,\ldots,m_k)\in M$ and $1\le i\le k$.

We conclude the proof by showing that the language $\iota[m_i]\circ \pi_i\circ\eta:\Sigma^*\to \Lambda$ is in $\bf L$ for each $m=(m_1,\ldots,m_k)\in M$ and $1\le i\le k$.
Let $\eta_{L_i}:\Sigma_i^*\to M_i$ be the syntactic morphism of $L_i$.
Because $\eta_{L_i}$ is surjective, there is a morphism $\psi:\Sigma^*\to \Sigma_i^*$ such that the diagram
    \begin{equation*}
  \xymatrix{
  {\Sigma^*}\ar@{..>}[r]^-{ \exists \psi}\ar@{->}[rd]_-{\pi_i\circ \eta}
  &{\Sigma_i^*}\ar@{->}[d]^-{\eta_{L_i}}&\\
  &{M_i}\ar@{->}[r]_-{\iota[m_i]}&\Lambda
  }
\end{equation*}
commutes by Proposition~\ref{prop-univ_free_2}.
Now, by Lemma~\ref{lem-disjunct_1}, the language $\iota[m_i]\circ \eta_{L_i}$ can be represented by using joins, meets, quotients, and $\Lambda$-morphisms with $L_i$.
This implies $\iota[m_i]\circ \eta_{L_i}\in \bf L$.
Therefore, its inverse image $(\iota[m_i]\circ \eta_{L_i})\circ \psi=\iota[m_i]\circ \pi_i\circ\eta$ under $\psi$ is also in $\bf L$.\qed
\end{proof}

\subsection{The Variety Theorem}
Let us prove the variety theorem. 
Let $\Lambda$ be a complete lattice.
For a positive variety $\bf L$ of $\Lambda$-languages, define the class of ordered monoids $\mathcal{V}({\bf L})$ as the smallest variety containing $\{M\mid M$ is the syntactic ordered monoid of $L\in {\bf L}\}$. 
In addition, for a pseudo-variety of ordered monoids $\cal M$, define the class of $\Lambda$-languages $\mathbf{V}({\cal M})$ as $\{L\mid L$ is recognized by $M\in {\cal M}\}$.
Note that $\mathbf{V}({\cal M})$ forms a positive variety by Theorem~\ref{th-basic_closure_properties}.

In \cite{pin1995variety}, the following is proved.
\begin{proposition}[Proposition 5.3, \cite{pin1995variety}]
    Let $\cal M$ be a pseudo-variety of ordered monoids, and let $(M,\le )\in \cal M$.
    Then there are languages $L_1,\ldots,L_k\in {\bf V}(\cal M)$ such that $(M,\le )$ is a submonoid of $M_{L_1}\times \cdots\times M_{L_k}$ where $M_{L_i}$ is the syntactic ordered monoid of $L_i$ with $1\le i\le k$.\bqed
\end{proposition}
Note that a $\mathbb{B}$-language can be considered as a $\Lambda$-language. Therefore, we have the following corollary.
\begin{corollary}\label{cor-subdirect_of_synt}
    Let $\cal M$ be a pseudo-variety of ordered monoids, and let ${(M,\le )}\in \cal M$.
    Then there are $\Lambda$-languages $L_1,\ldots,L_k\in {\bf V}(\cal M)$ such that $(M,\le )$ is a submonoid of $M_{L_1}\times \cdots\times M_{L_k}$.\bqed
\end{corollary}
Then, we can prove the main theorem.

\begin{theorem}[The Variety Theorem] For a complete lattice $\Lambda$,
$\mathcal{V}$ and $\mathbf{V}$ are mutually inverse. That is, $\mathbf{V}\circ \mathcal{V}({\bf L})={\bf L}$ and $ \mathcal{V}\circ\mathbf{V}({\cal M})={\cal M}$ for a positive variety ${\bf L}$ of $\Lambda$-languages and a pseudo-variety ${\cal M}$ of ordered monoids.
\end{theorem}
\begin{proof}
    We first show that $\mathbf{V}\circ \mathcal{V}({\bf L})={\bf L}$.
    The direction ${\bf L}\subseteq \mathbf{V}\circ \mathcal{V}({\bf L})$ is clear because the syntactic ordered monoid of $L$ recognizes $L$. 
    Let $M\in {\mathcal V}({\bf L})$ be an ordered monoid. 
    Then, $M$ is a divisor of the direct product $N$ of syntactic monoids of languages in ${\bf L}$.
    By Lemma~\ref{lem-recog_by_synt}, every language recognized by $N$ is in ${\bf L}$.
    In addition, if a language $L$ is recognized by $M$, then $N$ also recognizes $L$ by Theorem~\ref{th-basic_recognition}. Thus, $L$ is in ${\bf L}$.
    This implies $\mathbf{V}\circ \mathcal{V}({\bf L})\subseteq {\bf L}$.
    
    Next, we show that $\mathcal{V}\circ\mathbf{V}({\cal M})={\cal M}$.
    The direction ${\cal M}\subseteq \mathcal{V}\circ\mathbf{V}({\cal M})$ is implied by Corollary~\ref{cor-subdirect_of_synt}.
    Let $M_L\in \mathcal{V}\circ\mathbf{V}({\cal M})$ be the syntactic ordered monoid of a $\Lambda$-language $L\in \mathbf{V}({\cal M})$.
    By the definition of $\mathbf{V}({\cal M})$, $L$ is recognized by an ordered monoid $M\in {\cal M}$.
    By Theorem~\ref{th-minimality_of_synt}, $M_L$ is a divisor of $M\in {\cal M}$, and therefore, $M_L\in {\cal M}$. This implies $\mathcal{V}\circ\mathbf{V}({\cal M})\subseteq {\cal M}$.\qed
\end{proof}
We have considered a \emph{local} version of the variety theorem for lattices, where a lattice $\Lambda$ is fixed in advance. However, since the definition of pseudo-varieties of ordered monoids does not depend on the choice of $\Lambda$, we can conclude that for any two complete lattices $\Lambda$ and $\Lambda'$, there is a one-to-one correspondence between positive varieties of $\Lambda$-languages and those of $\Lambda'$-languages. 
Based on this observation, it is plausible that a \emph{global} version of the variety theorem, where $\Lambda$ is not fixed, can also be established. 
This is similar to the usual variety theorem, where it does not matter which alphabet $\Sigma$ is used. 
Intuitively, as inverse homomorphisms allow us to freely change the alphabet from $\Sigma$ to $\Sigma'$, $\Lambda$-morphisms allow us to freely change the lattice from $\Lambda$ to $\Lambda'$.

As noted in \cite{pin1995variety}, a variety of languages is also a positive variety. 
Although we considered ordered monoids for technical reasons, familiar varieties, such as groups or aperiodic monoids, are also covered by our result.

\subsection{Connection with Languages as Subsets}
%In this subsection, let us mention the relationship with usual languages that are not lattice languages for the case where $\Lamnda$ is finite.
In this subsection, we examine how positive varieties of $\Lambda$-languages relate to positive varieties of ordinary (i.e., non-lattice) languages when $\Lambda$ is finite.
To distinguish languages in the sense of subsets from lattice languages, we will use the hat symbol such as $\hat{L}\subseteq \Sigma^*$.

For a $\Lambda$-language $L:\Sigma^*\to \Lambda$, define $L_{\le \lambda}\subseteq \Sigma^*$ as  $\hat L_{\le\lambda}=\{w\mid L(w)\le \lambda\}$ for each $\lambda\in \Lambda$.
Then, we show that an ordered monoid recognizing $\hat L_{\le\lambda}$ for each $\lambda\in \Lambda$ has sufficient power to recognize $L$.
\begin{theorem}\label{th-link_with_usual_languages}
Let $\Lambda$ be a finite lattice, and $L:\Sigma^*\to \Lambda$ be a $\Lambda$-language.
Also, let $M_\lambda$ be an ordered monoid recognizing $\hat L_{\le\lambda}=\{w\mid L(w)\le \lambda\}\subseteq \Sigma^*$ for each $\lambda\in \Lambda$.
Then, $L$ is recognized by $\prod_{\lambda\in \Lambda}M_\lambda$.
\end{theorem}
\begin{proof}
We consider $\hat{L}_{\le \lambda}$ as the $\Lambda$-language $L_{\le \lambda}:\Sigma^*\to \Lambda$ such that $L_{\le \lambda}(w)=\mathbb{0}$ if $w\in \hat{L}_{\le \lambda}$, and $L_{\le \lambda}(w)=\mathbb{1}$ otherwise.
Then, there exists a triple $(\eta_\lambda,M_\lambda,P_\lambda)$ recognizing $L_{\le \lambda}$ by the assumption.
Define the op-coloring $P:\prod_{\lambda}M_\lambda\to \Lambda$ as \[
P((x_\lambda)_\lambda)=\bigwedge_{\lambda}\bigl(P_\lambda(x)\lor {\rm cons}(\lambda)\bigr)
\]
for each $x=(x_\lambda)_\lambda\in \prod_{\lambda}M_\lambda$.
It is clear that $P$ is order-preserving.
Then, for each $w\in \Sigma^*$,
\begin{align*}
    \bigl(P\circ \langle \eta_\lambda\rangle_\lambda \bigr) (w)&=P((\eta_{\lambda}(w))_\lambda)\\
    &=\bigwedge_{\lambda}\bigl(P_\lambda(\eta_{\lambda}(w))\lor {\rm cons}(\lambda)\bigr)\\
    &=\bigwedge_{\lambda}\bigl(L_{\le \lambda}(w)\lor {\rm cons}(\lambda)\bigr)
\end{align*}
holds. In addition, for each $\lambda\in \Lambda$,
\[
L_{\le \lambda}(w)\lor {\rm cons}(\lambda)=\begin{cases}\lambda&\text{if }L(w)\le\lambda,
\\
\mathbb{1}&\text{otherwise~.}\end{cases}
\]
Thus, \[
\bigl(P\circ \langle \eta_\lambda\rangle_\lambda \bigr) (w)=\bigwedge\{\lambda\mid L(w)\le \lambda\}=L(w)
\]
holds for each $w\in\Sigma^*$, that is, $L=P\circ \langle \eta_\lambda\rangle_\lambda $. 
Therefore, $\prod_{\lambda\in  \Lambda}M_\lambda$ recognizes $L$.\qed
\end{proof}

In \cite{pin1995variety}, shuffle ideals are introduced as an example of a class of languages characterized by ordered monoids.
Here, let us discussed the lattice version of shuffle ideals and its characterization.
For a word $w=a_1\cdots a_n$ with $a_1,\ldots,a_n\in \Sigma$, $w$ is a \emph{subword} of a word $v$ if $v=v_0a_1v_1\cdots v_{n-1}a_nv_n$ for some words $v_0,\ldots,v_n\in\Sigma^*$.
We write $w\sqsubseteq v$ if $w\in\Sigma^*$ is a subword of $v\in\Sigma^*$.
For a lattice $\Lambda$, a $\Lambda$-language $L:\Sigma^*\to \Lambda$ is a \emph{shuffle ideal} if $w\sqsubseteq v$ implies $L(v)\le L(w)$ for each $w,v\in\Sigma^*$.
By Theorem~\ref{th-link_with_usual_languages}, we have the following as an extension of Theorem~6.4 in \cite{pin1995variety}.
\begin{proposition}\label{prop-shuffle_ideal}
    Let $\Lambda$ be a finite lattice. 
    A $\Lambda$-language $L:\Sigma^*\to \Lambda$ is a shuffle ideal if and only if it is recognized by a finite ordered monoid $M$ such that the identity element $1_M\in M$ is the greatest element.
\end{proposition}
\begin{proof}
    For each $\lambda\in \Lambda$, define $\hat{L}_{\le \lambda}\subseteq \Sigma^*$ as  $\hat L_{\le\lambda}=\{w\mid L(w)\le \lambda\}$. 
    Then, $\hat L_{\le\lambda}$ is the shuffle ideal in the sense of \cite{pin1995variety} because if $w\sqsubseteq v$ and $L_{\le\lambda}(w)\le \lambda$, then $L_{\le \lambda}(v)\le L_{\le \lambda}(w)\le \lambda$ and $v\in \hat L_{\le\lambda}$.
    Thus, by Theorem~6.4 of \cite{pin1995variety}, each $\hat L_{\le\lambda}$ is recognized by an ordered monoid $M_\lambda$ such that $1_{M_\lambda}\in M_\lambda$ is the greatest element.
    By Theorem~\ref{th-link_with_usual_languages}, $L$ is recognized by $M=\prod_{\lambda}M_\lambda$, which is an ordered monoid such that $1_M=(1_{M_\lambda})_\lambda\in M$ is the greatest element.

    Conversely, assume that $L$ is recognized a triple $(\eta,M,P)$ such that $1_M\in M$ is the greatest element.
    Then, for a word $w=a_1\cdots a_n$ with $a_1,\ldots,a_n\in\Sigma$ and words $v_0,\ldots,v_n\Sigma^*$,
    we have $L(w)=P\circ \eta(w) =P\circ \eta(a_1\cdots a_n)=P( \eta(a_1)\cdots \eta(a_n))
    =P(1_M\cdot \eta(a_1)\cdot 1_M\cdots 1_M\cdot \eta(a_n)\cdot 1_M)
    \ge P(\eta(v_0)\eta(a_1)\eta(v_1)\cdots\allowbreak \eta(v_{n-1})\eta(a_n)\eta(v_n))
    =L(v_0a_1v_1\cdots v_{n-1}a_nv_n)$ by $1_M\ge v_i$ for each $0\le i\le n$. 
    Thus, $w\sqsubseteq v$ implies $L(v)\le L(w)$.\qed
\end{proof}

Of course, the class of shuffle ideals discussed here is just an example of positive varieties of lattice languages. 
As in this example, (positive) varieties of usual languages and their characterization can be extended to the lattice version in an appropriate way.

\section{Applications to Markov Chains}
In this section, we briefly discuss an applied aspect of our framework.
Specifically, we demonstrate that lattice languages and ordered monoids can provide a helpful method to analyze finite-state Markov chains.
Let us consider a Markov chain with a state space $\Omega=\{p_1,\ldots,p_n\}$ and a transition matrix $\Pi\in [0,1]^{\Omega\times \Omega}$.
A subset $C$ of $\Omega$ is called a \emph{communicating class} if for any two states $s_1,s_2\in C$, there exists a path from $s_1$ to $s_2$ with a positive probability.
Also, a subset $C\subseteq \Omega$ is called an \emph{ergodic class} if it is communicating and closed, i.e., $\Pi(s,t)=0$ for each $s\in C$ and $t\in\Omega \setminus C$
(see Chapter II of \cite{kemeny1960finite} in detail).
For example, the Markov chain $\mathcal{C}$ shown in Figure~\ref{fig-MC} has two ergodic classes $C_1=\{s_{11},s_{12}\}$ and $C_2=\{s_{21},s_{22}\}$.
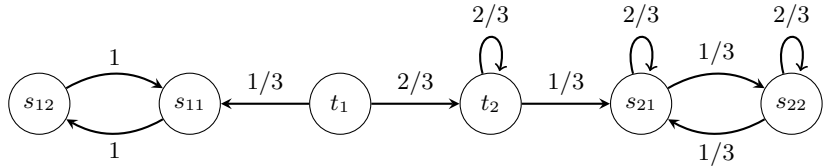
\begin{figure}[h]
      \centering
  \begin{tikzpicture}[on grid, node distance = 25pt,auto]
    \node [scale=1] (t1) [state] at (-2,0) {$t_1$};
    \node [scale=1] (s11) [state] at (-4,0) {$s_{11}$};
    \node [scale=1] (s12) [state] at (-6,0) {$s_{12}$};
    \node [scale=1] (t2) [state] at (0,0) {$t_2$};
    \node [scale=1] (s21) [state] at (2,0) {$s_{21}$};
    \node [scale=1] (s22) [state] at (4,0) {$s_{22}$};

    \path [-stealth, thick]
        (t1) edge node [above] {$1/3$} (s11)
        (s11) edge [bend left] node [below] {$1$} (s12)
        (s12) edge [bend left] node [above] {$1$} (s11)
        (t1) edge node [above]  {$2/3$} (t2)
        (t2) edge node [above] {$1/3$} (s21) 
        (s21) edge [bend left] node [above] {$1/3$} (s22)
        (s22) edge [bend left] node [below] {$1/3$} (s21)
      (t2) edge [loop above] node {$2/3$} ()
      (s21) edge [loop above] node {$2/3$} ()
      (s22) edge [loop above] node {$2/3$} ();
  \end{tikzpicture}
  \caption{Markov chain $\mathcal{C}$.}\label{fig-MC}
  \end{figure}
  
A Markov chain with a state space $\Omega$ is \emph{irreducible} if $\Omega$ itself is ergodic.
Because ergodic classes are convenient for analyzing Markov chains, reducible Markov chains are usually divided into ergodic classes and other components.
For example, the Markov chain $\mathcal{C}$ is divided into the class of transient states $\{t_1,t_2\}$ and ergodic classes $C_1,C_2$.

There have been studies that mention the relationship between Markov chains and formal languages such as \cite{etessami2009recursive,baier2023markov} and \cite{inoue2024semidirect}.
In particular, it is shown in \cite{rhodes2019unified} that every finite-state Markov chain has an action of a finite monoid on its state space that simulates the chain. 
Because an action of a finite monoid can be viewed as a DFA, every Markov chain can be converted into a DFA.
For example, $\mathcal{C}$ can be converted into the following DFA $\mathcal{A}$ with the alphabet $\Sigma=\{a,b,c\}$.
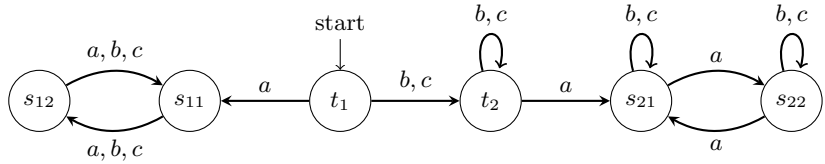
\begin{figure}[h]
      \centering
  \begin{tikzpicture}[on grid, node distance = 25pt,auto]
    \node [scale=1,initial above] (t1) [state] at (-2,0) {$t_1$};
    \node [scale=1] (s11) [state] at (-4,0) {$s_{11}$};
    \node [scale=1] (s12) [state] at (-6,0) {$s_{12}$};
    \node [scale=1] (t2) [state] at (0,0) {$t_2$};
    \node [scale=1] (s21) [state] at (2,0) {$s_{21}$};
    \node [scale=1] (s22) [state] at (4,0) {$s_{22}$};

    \path [-stealth, thick]
        (t1) edge node [above] {$a$} (s11)
        (s11) edge [bend left] node [below] {$a,b,c$} (s12)
        (s12) edge [bend left] node [above] {$a,b,c$} (s11)
        (t1) edge node [above]  {$b,c$} (t2)
        (t2) edge node [above] {$a$} (s21) 
        (s21) edge [bend left] node [above] {$a$} (s22)
        (s22) edge [bend left] node [below] {$a$} (s21)
      (t2) edge [loop above] node {$b,c$} ()
      (s21) edge [loop above] node {$b,c$} ()
      (s22) edge [loop above] node {$b,c$} ();
  \end{tikzpicture}
  \caption{DFA $\mathcal{A}$ simulating $\mathcal{C}$.}
  \end{figure}
  
\noindent
In this context, it is natural to ask the \emph{problem of ergodic classes}: Which ergodic class does each word reach?
To consider this, let us use lattice languages.
First, let $\Lambda=\mathcal{P}(\{1,2\})$ be the lattice of the powerset of $\{1,2\}$. 
That is, $\emptyset=\mathbb{0}$ is the least element, $\{1,2\}=\mathbb{1}$ is the greatest element, and $\{1\}$ and $\{2\}$ are incomparable. 
Next, let $F:\Omega\to \Lambda$ be the coloring on the DFA $\mathcal{A}$ defined as\[
F(s)=\begin{cases}
    \{1\}&\text{ if }s\in C_1=\{s_{11},s_{12}\},\\
    \{2\}&\text{ if }s\in C_2=\{s_{21},s_{22}\},\\
    \{1,2\}&\text{ otherwise}.
\end{cases}
\]
The $\Lambda$-language $L:\Sigma^*\to \Lambda$ is defined in the usual way, e.g., $L(ab)=\{1\}$ and $L(bbc)=\{1,2\}$.
Then, of course, $L^{-1}(\{i\})=\{w\in\Sigma^*\mid w\text{ 
reaches a state in }C_i\}$ holds for each $i\in\{0,1\}$.
In addition, the probability that a randomly chosen word reaches an ergodic class $C_i$ can be computed as the \emph{probability} of the language $L^{-1}(\{i\})$ (see e.g. Chapter III of \cite{salomaa2012automata} in detail).

More importantly, it is easily shown that this $\Lambda$-language $L$ is a shuffle ideal (see Section~3.4).
Therefore, we immediately obtain that the syntactic ordered monoid $M_L$ of $L$ is an ordered monoid in which the identity element is the greatest element by Proposition~\ref{prop-shuffle_ideal}.
In general, any language calculating the problem of ergodic classes is a shuffle ideal.
Consequently, various results from language theory can be applied to these languages and ordered monoids.
For example, because shuffle ideals are star-free (e.g., Chapter VII of \cite{pin2010mfat}), syntactic ordered monoids are always aperiodic by Schützenberger's theorem \cite{schutzenberger1965finite}.
This is also consistent with the fact that a random walk on a group cannot be equivalent to any reducible Markov chain.\footnote{Groups and aperiodic monoids are orthogonal: only the trivial monoid is both invertible and aperiodic.}

In this section, for simplicity, we have only discussed ergodic classes. 
However, more generally, transient states and its hierarchy can be effectively described using shuffle ideals. For instance, in $\mathcal{C}$, even if we modify the lattice $\Lambda$ and coloring $F$ so that $F(t_2)\lneq F(t_1)\in\Lambda$ by reflecting the fact that $t_2$ cannot return to $t_1$, $L$ is still a shuffle ideal.
Further discussion is omitted due to space limitations, but the authors believe that various other discussions in Markov chains can be effectively described using lattice languages and ordered monoids.

\section{Conclusion}
In this paper, we proved a natural one-to-one correspondence between positive varieties of lattice languages and pseudo-varieties of ordered monoids, as an extension of the result in \cite{pin1995variety}.
Additionally, we outlined how our framework can be applied to the algebraic approach to Markov chains.
Here, we explain two directions for future work.

One of the future directions is to engage in a deeper discussion regarding the application to Markov chains. 
The description of reducibility and the relationship with probabilities in language theory were briefly discussed in Section~4. 
Additionally, there is an important unsolved problem known as the \emph{cutoff phenomenon} \cite{diaconis1981generating,aldous1986shuffling}.
The cutoff phenomenon refers to the situation where a Markov chain quickly approaches its stationary distribution after a specific point in time. 
This is empirically known to occur in Markov chains with certain symmetries, such as random walks on groups. 
Therefore, it is possible that our approach, based on pseudo-varieties, could provide insights into this problem. 
%Moreover, there are many future studies on the relationship between Markov chains and our framework.

In addition, there are subsequent studies of \cite{pin1995variety} such as Pol{\'a}k's correspondence between conjunctive varieties of languages and pseudo-varieties of idempotent semirings \cite{Polak}, and Ballester-Bolinches's correspondence between formations of languages and formations of monoids \cite{BallesterBolinches}.
It is an interesting question whether these variants can also be extended to lattice versions. 
%In particular, the problem of ergodic classis discussed in Section~4 seems to be better described using semilattices rather than lattices.
\newpage
\bibliography{bibtex}

\end{document}